\documentclass[a4paper,11pt]{article}

\usepackage{comment}
\usepackage[textsize=tiny,textwidth=2cm,shadow]{todonotes}

\usepackage{typearea}
\typearea{12}

   \usepackage{url,amsmath,amsthm,amssymb}
   \usepackage{color}
   \usepackage{tikz,framed}
   \usepackage[textwidth=2cm]{todonotes}
 
\usepackage[linesnumbered,ruled,lined,boxed]{algorithm2e}
\usepackage{amsmath,amsthm,amssymb}
\usepackage{tikz,framed}
\usepackage[linesnumbered,ruled,lined,boxed]{algorithm2e}
\newtheorem{theorem}{Theorem}[section]
\newtheorem{lemma}{Lemma}[section]
\newtheorem{corollary}{Corollary}[section]

\newtheorem{proposition}{Proposition}[section]

\newtheorem{example}{Example}[section]

\newtheorem{remark}{Remark}[section]

\newcommand{\N}{\mathbb{N}}

\DeclareMathOperator*{\argmax}{arg\,max}

\newcommand{\R}{\mathbb{R}}
\newcommand{\T}{\mathcal{T}}
\renewcommand{\L}{\mathcal{L}}

\title{Toll Caps in Privatized Road Networks}

\author{Tobias Harks\thanks{University of Augsburg, Augsburg, Germany, tobias.harks@math.uni-augsburg.de}
 \and
Marc Schr\"{o}der
\thanks{%
  RWTH Aachen University, Aachen, Germany, marc.schroeder@oms.rwth-aachen.de}
   \and Dries Vermeulen
\thanks{
  Maastricht University, Maastricht, The Netherlands, d.vermeulen@maastrichtuniversity.nl}
}

\begin{document}
\maketitle
\begin{abstract}
We consider a network pricing game on a parallel network with congestion effects in which link owners set tolls for travel so as to maximize profit. A central authority is able to regulate this competition by means of a (uniform) price cap. The first question we want to answer is how such a cap should be designed in order to minimize the total congestion. We provide an algorithm that finds an optimal price cap for networks with affine latency functions and a full support Wardrop equilibrium. Second, we consider the induced network performance at an optimal price cap. We show that for two link networks with affine latency functions, the congestion costs at the optimal price cap are at most $8/7$ times the optimal congestion costs. For more general latency functions, this bound goes up to 2 under the assumption that an uncapped Nash equilibrium exists. However, in general such an equilibrium need not exist and this can be used to show that optimal price caps can induce arbitrarily inefficient flows.
\end{abstract}
\medskip

\noindent \emph{ Keywords:}\ Game theory, competition regulation, toll caps, Nash equilibrium, Wardrop equilibrium.

\section{Introduction}
With the ongoing privatization of public road infrastructure,
toll charging on highways and roads is common practice
in many cities around the world (see e.g. Bergen~\cite{Bergen07},
London~\cite{London08}, Santiago de Chile~\cite{Santiago16}, Singapore~\cite{Singapur10}, and Stockholm~\cite{Stockholm07}). 
The toll market in the United States (starting already in the 18th century)
is built on the key idea that private firms obtain the right to construct the infrastructure (usually via an auction, see Porter and Zona~\cite{Porter93})
and as compensation are allowed to charge tolls for road usage.
Firms are further obligated to reinvest parts of the revenues to maintain the infrastructure.
To date about 35 States use this mechanism and while it has worked well in some cases, several problems arose in the past and still prevail. We report here one incident. In 2012,
large toll rate increases have been implemented by the Port Authority of New York and New Jersey (justified in part to finance its World Trade Center project).
In response, a bill was introduced in Congress that would allow the Secretary of Transportation to \emph{regulate} tolls on every bridge on the country's interstates and other federally aided highways. 

If a regulation authority introduces toll caps, how should this cap be designed
in order to induce a socially beneficial outcome and how does this outcome perform?
This question is at the core of our paper. 
We model the situation using a three-level optimization model, where
in the lowest (third) level, there are commuters that want to travel
from an origin to a destination and each commuter
minimizes a combination of latency costs plus toll costs.
We model the route choice of a fixed volume of commuters using the concept of Wardrop equilibrium~\cite{Wa52}.
We further assume that the underlying network consists of $n\geq2$ parallel edges connecting a common
origin with a common destination. While this assumption is indeed restrictive, it
still models relevant situations, for instance, when there are parallel access roads (e.g.
 tollable bridges) connecting the central district of a city with a suburb, and it is used in recent related work (e.g. Acemoglu and Ozdaglar~\cite{AcOz07}, and Wan~\cite{Wan16}).
 
In the middle (second) layer, there are firms owning the edges and charging tolls on them. Each firm maximizes revenue, but faces competition from other
firms. We assume that each firm owns only one edge. A \emph{Nash equilibrium} is a vector of tolls, one toll value per firm, so that no firm can improve by unilaterally
deviating to another toll value. A Nash equilibrium takes into account that commuters will (potentially) change their route choices once a firm deviates to a different toll.\footnote{In game theory this equilibrium concept is called subgame perfect Nash equilibrium, but given that the focus of this paper is on the pricing game, we will call it a Nash equilibrium.}

The authority in the first and highest level can decide on some
toll cap. We assume that the authority is not allowed to discriminate between the different firms and puts one (uniform) price cap applicable to all firms. We consider the situation in which the authority is purely
interested in minimizing the resulting total latency costs (clearly
other objectives may be feasible as well).

\subsection{Our Results}
We study a three-level optimization problem in which a central authority can impose a (uniform) price cap so as to minimize the total induced equilibrium congestion costs. We motivate this form of competition regulation by showing that without any price regulation, the induced equilibrium congestion costs can be arbitrarily higher than the optimal congestion costs because of too high prices of one of the firms. 

For solving the three-level optimization problem, we first derive a set of structural results for instances with affine latency functions. In particular, we show that a Nash equilibrium exists for all price caps $c$, and is in fact unique whenever the Wardrop equilibrium has full support. 
Moreover, under the full support assumption, we derive a complete characterization
of  Nash equilibria using the KKT conditions of the underlying optimization
problems of the involved firms.
Based on this characterization, we devise an algorithm that finds the optimal price cap in polynomial time. The algorithm computes  breakpoints at which the set of firms putting their price equal to the cap changes. We show that these sets only increase (with decreasing cap) so that there are at most $n$  breakpoints, where $n$ denotes the number of firms. The breakpoints essentially divide the real line into at most $n$ intervals and for each of these intervals we show that the cost function is quadratic. Thus, the algorithm needs to solve at most $n$ quadratic one-dimensional minimization problems, which can be done in polynomial time.

We then consider the performance of an optimal price cap. Given that a Nash equilibrium always exists for networks with affine latency functions, we show a tight bound of $8/7$ between the congestion costs at an optimal price cap and the optimal congestion costs for duopoly instances. For more general latency functions, this bound goes up to 2 if we assume that an uncapped Nash equilibrium exists.

In general, however, the objective function of a firm in the second level is not concave, the set of best replies is not convex and therefore standard techniques for proving existence of a Nash equilibrium fail. In fact, we give an example in which an uncapped subgame perfect Nash equilibrium does not exist. Moreover, we are able to use the nonexistence of uncapped Nash equilibria to construct a sequence of instances in which the ratio between the congestion costs at the optimal price cap and the optimal congestion costs goes to infinity.

\subsection{Related Work}
The inefficiency of selfish behavior in congestion games has been well recognized since the work of Pigou~\cite{Pi20} in economics and Wardrop~\cite{Wa52} in transportation networks. A natural way to quantify this inefficiency is the price of anarchy as defined by Koutsoupias and Papdimitrou~\cite{KoPa99}. For routing games this is first done by Roughgarden and Tardos~\cite{RoTa02} and later by Correa et al.~\cite{CorSS04}.

One way to restore efficiency is by means of centralized pricing. Beckmann et al.~\cite{BeMcWi56} and Dafermos and Sparrow~\cite{DaSp69} showed that marginal cost pricing induces an equilibrium flow that is optimal. Despite the effectiveness of marginal cost tolls, there are two drawbacks. First, potentially each edge in the network is tolled; an issue that is considered by Hoefer et al.~\cite{HoOlSk08}
and Harks et al.~\cite{HKKM2011}. Second, the imposed tolls can be arbitrary large; an issue that is considered by Bonifaci et al.~\cite{BoSaSc11}, Fotakis et al.~\cite{FoKaLi15}, Jelinek et al.~\cite{JeKlSc14} and Kleer and Sch\"{a}fer~\cite{KlSc16}. Cole et al.~\cite{CoDoRo03} show that optimal tolls also exist when users are heterogeneous with respect to the tradeoff between time and money. Later this result was extended to general networks by Fleischer et al.~\cite{flJaMa04}, Karakostas and Kolliopoulos~\cite{KaKo04} and Yang and Huang~\cite{YaHu04}.

Instead of improving efficiency, tolls can also be used to minimize or maximize the profit of one or multiple leaders. The model that computes tolls that induce the optimal flow at minimal
profit was analyzed by Dial~\cite{DialpartI,DialpartII}. The model with one profit maximizing leader and no congestion effects, was first analyzed by Labb\'{e} et al.~\cite{LaMaSa98} and later by Briest et al.~\cite{BrHoKr12}.

Acemoglu and Ozdaglar~\cite{AcOz07} introduced a model of price competition between link owners in a parallel network that is very similar to ours. The main difference is that in their model users have a reservation value for travel. This implies that if links become too expensive users choose not to travel, whereas in our model users always travel through the network. Their main result is a (tight) bound on the inefficiency of equilibria. 
Later several generalizations of the model were introduced. The follow-up work of Acemoglu and Ozdaglar~\cite{AcOz:07} allows for a slightly more general topology, namely parallel paths with multiple links, and show that equilibria can be arbitrarily inefficient. Hayrapetyan et al.~\cite{HaTaWe07} consider the model with elastic traffic demand. Their bounds on inefficiency, however, are not tight, which is improved upon by Ozdaglar~\cite{Oz08}.

Recently, following our question on competition regulation of price competition between edge owners, Correa et al.~\cite{CoGuLiNiSc18} considered the setting in which a central authority is allowed to put different toll caps on different edges of the network. Their main result is that for all network topologies there are toll caps so that firms are willing to put their toll equal to the cap, and the induced equilibrium flow is the optimal flow. An example of such caps are marginal tolls as introduced by Beckmann et al.~\cite{BeMcWi56}. Notice that the assumption of individual toll caps is important. In practice, however, toll cap discrimination is often not allowed, which motivates our work on uniform toll caps. 

A different extension was proposed by Johari et al.~\cite{Jo10}. They extend the analysis of price competition by including investment and entry decisions. The model of Xiao et al.~\cite{XiYaHa07} studies competition in both tolls and capacities and finds that tolls are higher, but capacities are lower than socially desired. Other recent models of Bertrand competition in a network setting that use different ways of modelling congestion effects are Anshelevich and Sekar~\cite{AnSe15}, Chawla and Niu~\cite{ChNi09}, Chawla and Roughgarden~\cite{ChRo08}, Papadimitriou and Valiant~\cite{PaVa10} and Caragiannis et al.~\cite{CaChKaKrPrVo17}.

\section{The model}
An instance of the three-level optimization problem is
given by the tuple $I=(N,(\ell_i)_{i\in N})$, where  $N=\{1,\ldots,n\}$, with $n\geq 2$, is a set of parallel
links connecting a common source with a common destination. There is one unit of flow 
to be sent over the links. Let $x_i\in\R_+$ be the total flow on link $i$ and $x=(x_1,\ldots,x_n)$ be the vector of flows. Each link has a flow-dependent latency function $\ell_i(x_i)$ that we assume to be strictly increasing, convex and continuously differentiable. We denote this class of functions by $\L_c$. By $\L_d$ we denote the class of strictly increasing polynomial latency functions with nonnegative coefficients and degree at most $d$. In particular, $\L_1$ denotes the class of strictly increasing affine latency functions.

Given a flow $x$, define the total latency costs by
$$C(x)=\sum_{i\in N}\ell_i(x_i)\cdot x_i.$$
A flow $x^*$ is \textit{optimal}, if $C(x^*)\leq C(x)$ for all flows $x$.

Let $t_i\in\mathbb{R}_+$ be the toll on link $i$ and $t=(t_1,\ldots,t_n)$ be the vector of tolls. The \textit{effective cost} of a user of link $i$ is $\ell_i(x_i)+t_i$. Given a toll vector $t$, a flow $x$ is a \textit{Wardrop equilibrium for $t$}, if for all $i,j\in N$ with $x_i>0$, $$\ell_i(x_i)+t_i\leq \ell_j(x_j)+t_j.$$

In particular, if $x$ is an equilibrium for $t$, then all links with positive flow have equal effective costs, i.e., there is some $K>0$ with $\ell_i(x_i)+t_i=K$ for all $i\in N$ with $x_i>0$. It is well-known that given our assumptions on the latency functions, an equilibrium for $t$ exists, is unique and can be described by means of the following inequality (see Beckmann et al.~\cite{BeMcWi56} and Dafermos and Sparrow~\cite{DaSp69}).
\begin{lemma}\label{lem:vi}
A flow $x$ is a Wardrop equilibrium for $t$ if and only if for all feasible flows $x'$,
$$\sum_{i\in N}(\ell_i(x_i)+t_i)\cdot(x_i-x'_i)\leq 0.$$
\end{lemma}

For each toll vector $t$, we denote by $x(t)$ the unique equilibrium for $t$. The equilibrium for $t=0$ is called the \textit{Wardrop equilibrium}.

\subsection{Nash Equilibrium}
We assume that each link $i\in N$ is owned by a firm (also denoted by $i$)
and that the objective of each firm is to maximize profit. Given a toll vector $t\in\mathbb{R}^N_+$, define $\Pi_i(t_i,t_{-i})=t_i\cdot x_i(t)$ for all $i\in N$. A toll vector $t$ is a \textit{$c$-capped Nash equilibrium}, if for all $i\in N$, $t\leq c$ and for all $0\leq t'_i\leq c$,
$$\Pi_i(t_i,t_{-i})\geq \Pi_i(t'_i,t_{-i}).$$
Observe that the flow adapts to the toll vector $(t'_i,t_{-i})$ in the calculation of $\Pi_i(t'_i,t_{-i})$. Given $t_{-i}\in\mathbb{R}^{N\setminus\{i\}}_+$ and $c\in\mathbb{R}_+$, define $B^c_i(t_{-i})=\argmax\limits_{0\leq t_i\leq c}\Pi_i(t_i,t_{-i})$ and $B^{\infty}_i(t_{-i})=\argmax\limits_{t_i\geq 0}\Pi_i(t_i,t_{-i})$ for all $i\in N$.

Given $c\in\mathbb{R}_+$, we denote by
\[ \T(c)=\{t\in \R_+^n \mid t \text{ is a $c$-capped Nash equilibrium}\}\]
the set of $c$-capped Nash equilibria. By $\T(\infty)$ we denote the set of uncapped Nash equilibria. Note that there might be instances for which $\T(c)$ is not a singleton for some $c\in\R_+$, and for which $\T(c)=\emptyset$ for some $c\in\R_+$.

\subsection{Designing Price Caps}

The following example provides the motivation for our research: in terms of latency costs, a Nash equilibrium can be arbitrarily worse than the optimal flow.
\begin{example}
Consider the network of Figure \ref{fig:bad}.
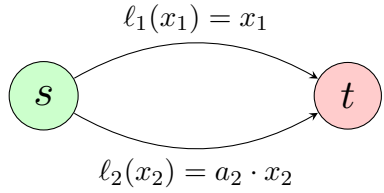
\begin{figure}[ht]
\centering
\begin{tikzpicture}
\node[draw,circle,scale=1.5,fill=green!20] (1) at (0,0) {$s$};
\node[draw,circle,scale=1.4,fill=red!20] (2) at (4,0) {$t$};
\draw[-stealth] (1)[out=30,in=150] to node[above]{$\ell_1(x_1)=x_1$} (2);
\draw[-stealth] (1)[out=-30,in=-150] to node[below]{$\ell_2(x_2)=a_2\cdot x_2$} (2);
\end{tikzpicture}
\caption{An inefficient Nash equilibrium as $a_2\rightarrow \infty$.}
\label{fig:bad}
\end{figure}

Then $$B^{\infty}_1(t_2)=\begin{cases}
\frac{a_2+t_2}{2},&\text{ if }t_2\leq a_2+2,\\
t_2-1,&\text{ if }t_2>a_2+2,
\end{cases}$$ and
$$B^{\infty}_2(t_1)=\begin{cases}
\frac{1+t_1}{2},&\text{ if }t_1\leq 2a_2+1,\\
t_1-a_2,&\text{ if }t_1>2a_2+1.
\end{cases}$$
Combining $B^{\infty}_1(t_2)$ and $B^{\infty}_2(t_1)$, if $t\in\T(\infty)$, then $t=\left(\frac{2a_2+1}{3},\frac{a_2+2}{3}\right)$. Since $x(t)=\left(\frac{2a_2+1}{3a_2+3},\frac{a_2+2}{3a_2+3}\right)$ and $x^*=\left(\frac{a_2}{a_2+1},\frac{1}{a_2+1}\right)$, we have $C(x(t))=\frac{a_2^2+7a_2+1}{9a_2+9}$ and $C(x^*)=\frac{a_2}{a_2+1}$. Hence $\frac{C(x(t))}{C(x^*)}=\frac{a_2^2+7a_2+1}{9a_2}\rightarrow\infty$ as $a_2\rightarrow \infty$. As $a_2\rightarrow \infty$, the price of firm 1 grows faster than the price of firm 2, which pushes too much flow onto link 2, creating inefficiencies.
\end{example}

The bad performance of the Nash equilibrium calls for competition regulation. The regulation policy we focus on is that of a price cap $c\in\R_+$: no firm is allowed to set its toll above the price cap, i.e., $t_i\leq c$ for all $i\in N$.

Formally, we seek to choose $c$ so as to solve the following three-level optimization problem:

\begin{framed}
\begin{equation}\label{eq:problem}
\tag{3L-P}\inf_{c\in\R_+}\sup_{t\in \T(c)} C(x(t)).
\end{equation}
\end{framed}
If $\T(c)=\emptyset$ for some $c\in \R_+$, we assume that $C(x(t))=\infty$. Observe that $\T(0)\neq\emptyset$.

Note that this is a quite appealing and robust formulation: we minimize total costs for the \emph{worst} possible realization
of a Nash equilibrium. It is known, however, that multi-level optimization
problems are notoriously hard in terms
of proving existence and computability of optimal solution.

In the next section, we derive structural results of Nash equilibria 
and their corresponding Wardrop equilibria. These structural results contain
existence and uniqueness of equilibria for networks with affine latencies
and a comparative statics result with respect to different price caps.
These insights form the basis of our polynomial time
algorithm computing an optimal price cap for the case of affine latency functions.

\section{Structural Properties of Nash Equilibria}
Assume that $\ell_i(x_i)=a_i\cdot x_i+b_i$ with $a_i> 0$ and $b_i\geq 0$. The first result shows that for affine latency functions, a Nash equilibrium exists for all price caps. A similar result is obtained for the model in which users have a reservation value for travel and there are no price caps by Acemoglu and Ozdaglar~\cite{AcOz07} and Hayrapetyan et al.~\cite{HaTaWe07}.

\begin{proposition} \label{pro:exi}
Let $\ell_i\in\mathcal{L}_1$ for all $i\in N$. Then, there exists a $c$-capped Nash equilibrium for all $c\in\mathbb{R}_+$.
\end{proposition}
\begin{proof}
%[Proof of Proposition \ref{pro:exi}]
Assume that $\ell_i(x_i)=a_i\cdot x_i+b_i$ with $a_i> 0$ and $b_i\geq 0$. For all toll vectors $t\in\mathbb{R}^N_+$, define $N(t)=\{i\in N\mid x_i(t)>0\}$.

Since $x(t)$ is an equilibrium for $t$, we have $a_i\cdot x_i(t)+b_i+t_i=K$ for all $i\in N(t)$. With $\sum_{i\in N(t)}x_i(t)=1$, we get
$K=\frac{1+\sum_{j\in N(t)}\frac{b_j+t_j}{a_j}}{\sum_{j\in N(t)}\frac{1}{a_j}}$
and thus
\begin{equation}
x_i(t)=\frac{1+\sum\limits_{j\in N(t)}\frac{b_j+t_j-b_i-t_i}{a_j}}{\sum\limits_{j\in N(t)}\frac{a_i}{a_j}}\text{ for all }i\in N(t). \label{eq:flo}
\end{equation}

By the Theorem of the Maximum \cite{Be63}, it follows that $B^c(t)=(B^c_i(t_{-i}))_{i\in N}$ is upper semi-continuous and hence has a closed graph. We show that $B^c_i(t_{-i})$ is a convex set for all $t_{-i}$.

Assume $t_{-i}\in\mathbb{R}_+^{N\setminus\{i\}}$ and consider the following two cases: (a) $\Pi_i(t_i,t_{-i})=0$ for all $t_i\in\mathbb{R}_+$ and (b) $\Pi_i(t_i,t_{-i})>0$ for some $t_i\in\mathbb{R}_+$.

Case (a): Suppose $\Pi_i(t_i,t_{-i})=0$ for all $t_i\in\mathbb{R}_+$. Obviously $B^c_i(t_{-i})$ is convex.

Case (b): Suppose there is some $t'_i\in\mathbb{R}_+$ with $\Pi_i(t'_i,t_{-i})>0$. We will prove that $\Pi_i(t_i,t_{-i})$ is a concave function in $t_i$ until there is some $\bar{t}_i$ such that $\Pi_i(t_i,t_{-i})=0$ for all $t_i\geq\bar{t}_i$. Observe the following from \eqref{eq:flo}.

\begin{enumerate}
\item[(i)] If $j\in N(t'_i,t_{-i})$, with $j\neq i$, for some $t'_i\in\mathbb{R}_+$, then $j\in N(t_i,t_{-i})$ for all $t_i\geq t'_i$. That is, the set $N(t)$ increases as $t_i$ increases.
\item[(ii)] For a fixed $N(t)$, $\Pi_i(t_i,t_{-i})$ is either linear or quadratic in $t_i$. $\Pi_i(t)$ is linear if and only if $N(t)=\{i\}$, which can only happen for $t_i\in[0,\hat{t}_i]$ for some $\hat{t}_i\in\mathbb{R}_+$. In all other cases, $\Pi_i(t)$ is quadratic with a decreasing slope. 
\end{enumerate}
Now in order to prove that $\Pi_i(t_i,t_{-i})$ is a concave function in $t_i$, it is sufficient to show that the slope decreases at every toll $t_i$ for which new links start to receive flow, i.e. a new set of players joins $N(t_i,t_{-i})$. To this end, let $t_i$ denote a price at which a new set of firms starts to receive flow and let us denote this set by $\bar{N}$. Notice that 
\begin{align*}
\frac{\partial_-\Pi_i(t_i,t_{-i})}{\partial t_i}&=\frac{1+\sum\limits_{j\in N(t):j\neq i}\frac{b_j+t_j-b_i-2t_i}{a_j}}{\sum\limits_{j\in N(t)}\frac{a_i}{a_j}}\\
\frac{\partial_+\Pi_i(t_i,t_{-i})}{\partial t_i}&=\frac{1+\sum\limits_{j\in N(t)\cup \bar{N}:j\neq i}\frac{b_j+t_j-b_i-2t_i}{a_j}}{\sum\limits_{j\in N(t)\cup\bar{N}}\frac{a_i}{a_j}}
\end{align*}
where $\frac{\partial_-\Pi_i(t_i,t_{-i})}{\partial t_i}$ and $\frac{\partial_+\Pi_i(t_i,t_{-i})}{\partial t_i}$ denote the left and right partial derivative with respect to $t_i$, respectively.
So in order to prove that the slope decreases it is sufficient to show that
\begin{equation}\label{eq:slo}
\left(1+\sum\limits_{j\in N(t):j\neq i}\frac{b_j+t_j-b_i-2t_i}{a_j}\right)\cdot\sum\limits_{k\in \bar{N}}\frac{a_i}{a_k}\geq\sum\limits_{k\in \bar{N}}\frac{b_k+t_k-b_i-2t_i}{a_k}\cdot{\sum\limits_{j\in N(t)}\frac{a_i}{a_j}}.
\end{equation}
We get
\begin{align*}
&\sum\limits_{k\in \bar{N}}\frac{b_k+t_k-b_i-2t_i}{a_k}\cdot{\sum\limits_{j\in N(t)}\frac{a_i}{a_j}}\\
&=\sum\limits_{k\in \bar{N}}\frac{a_i\cdot x_i(t)-t_i}{a_k}\cdot{\sum\limits_{j\in N(t)}\frac{a_i}{a_j}}\\
&=\sum\limits_{k\in \bar{N}}\frac{1+\sum\limits_{j\in N(t):j\neq i}\frac{b_j+t_j-b_i-t_i}{a_j}-\sum\limits_{j\in N(t)}\frac{t_i}{a_j}}{a_k\cdot\sum\limits_{j\in N(t)}\frac{1}{a_j}}\cdot\sum_{j\in N(t)}\frac{a_i}{a_j}\\
&=\sum\limits_{k\in \bar{N}}\left(1+\sum\limits_{j\in N(t):j\neq i}\frac{b_j+t_j-b_i-2t_i}{a_j}\right)\cdot\frac{a_i}{a_k}-\frac{t_i}{a_k}\\
&=\left(1+\sum\limits_{j\in N(t):j\neq i}\frac{b_j+t_j-b_i-2t_i}{a_j}-\frac{t_i}{a_i}\right)\cdot\sum\limits_{k\in \bar{N}}\frac{a_i}{a_k},
\end{align*}
where the first equality follows from $b_k+t_k=a_i\cdot x_i(t)+b_i+t_i$ for all $k\in \bar{N}$, the second equality from the definition of $x_i(t)$, and the last two inequalities from rewriting. Since $\frac{t_i}{a_i}>0$, \eqref{eq:slo} holds true. Hence $B^c_i(t_{-i})$ is convex-valued. 

Kakutani's fixed point theorem~\cite{Ka41} now implies that there exists an equilibrium for all price caps $c\in\R_+$.
\end{proof}

The following lemma states a very natural comparative statics result,
showing that a unilateral increase in toll by one firm only
decreases the flow on the corresponding link while the flow
on other links only increases.
\begin{lemma}\label{lem:mon}
Let $i\in N$ and $t=(t_i,t_{-i}), t'=(t'_i,t_{-i})\in \R_+^n$ with $t_i\leq t'_i$. Then $x_i(t)\geq x_i(t')$ and $x_j(t)\leq x_j(t')$ for all $j\neq i$.
\end{lemma}

\begin{proof}
It is well known (see Beckmann et al.~\cite{BeMcWi56} and Dafermos and Sparrow~\cite{DaSp69}) that $x$ is an equilibrium for $t$ if $x$ solves the following minimization problem
$$
\min_x \sum_{i\in N}\int_{0}^{x_i} (\ell_i(y)+t_i)\,dy.
$$
Hence
\begin{eqnarray*}
\sum_{i\in N}\int_{0}^{x_i(t)} (\ell_i(y)+t_i)\,dy &\leq& \sum_{i\in N}\int_{0}^{x_i(t')} (\ell_i(y)+t_i)\,dy \\
\sum_{i\in N}\int_{0}^{x_i(t')} (\ell_i(y)+t'_i)\,dy &\leq& \sum_{i\in N}\int_{0}^{x_i(t)} (\ell_i(y)+t'_i)\,dy \\
\end{eqnarray*}
Combining these inequalities we get
\begin{equation}\label{eq:monotonicity}
\sum_{j\in N}(t_j-t'_j)\cdot(x_j(t)-x_j(t'))\leq 0.
\end{equation}
Using $t'=(t'_i,t_{-i})$, we get
$$(t_i-t'_i)\cdot(x_i(t)-x_i(t'))\leq 0,$$
and thus $x_i(t)\geq x_i(t')$. 

Before we prove that $x_j(t)\leq x_j(t')$ for all $j\neq i$, we first show that $\ell_i(x_i(t))+t_i\leq \ell_i(x_i(t'))+t'_i$. Since \(x(t)+(x(t^{\prime})-x(t))=x(t^{\prime})\) 
is a feasible flow, \(x(t^{\prime})-x(t)\) is a feasible direction for \(x(t)\). By the 
first-order optimality conditions,
\[\sum_{i\in N} (\ell_i(x_i(t))+t_i)\cdot(x_i(t)-x_i(t^{\prime})) \leq 0.\] 
Analogously,
\(x(t)-x(t^{\prime})\) is a feasible direction for \(x(t^{\prime})\), and thus
\[ \sum_{i\in N} (\ell_i(x_i(t^{\prime}))+t^{\prime}_i)\cdot(x_i(t^{\prime})-x_i(t)) \leq 0.\]
Adding up these inequalities, we obtain
\begin{align*}
&\sum_{j\neq i} (x_j(t)-x_j(t'))\cdot(\ell_j(x_j(t))-\ell_j(x_j(t')))\\
&+ (x_i(t)-x_i(t'))\cdot(\ell_i(x_i(t))+t_i-\ell_i(x_i(t'))-t'_i) \leq 0.
\end{align*} 
Notice that the first summation term 
is nonnegative, as \(\ell_i\) is increasing for all \(i\in N\). Thus, \begin{equation}\ell_i(x_i(t))+t_i\leq \ell_i(x_i(t'))+t'_i. \label{eq:mon}
\end{equation}

Now we prove that $x_j(t)\leq x_j(t')$ for all $j\neq i$. Notice that we can assume that $x_i(t)>0$, as otherwise the conclusion follows trivially. We use a proof by contradiction. Suppose that $x_j(t)>x_j(t')$ for some $j\neq i$. Then
\begin{align*}
\ell_j(x_j(t'))+t_j<\ell_j(x_j(t))+t_j\leq \ell_i(x_i(t))+t_i\leq \ell_i(x_i(t'))+t'_i\leq \ell_j(x_j(t'))+t_j,
\end{align*}
where the first inequality follows from $x_j(t)>x_j(t')$, the second inequality from $x_j(t)>0$, the third inequality from \eqref{eq:mon}, and the fourth inequality from $x_i(t')>0$. This is a contradiction and finishes the proof.
\end{proof}

From now on, we assume that $x_i(0)>0$ for all $i\in N$. In terms
of practical applications, this assumption seems not overly
restrictive since a road segment is usually
only priced if there is traffic without any tolls. However, the assumption is important for proving uniqueness of Nash equilibria. Example~\ref{ex:mul} shows that multiple Nash equilibria may exist in case the full Wardrop support assumption is not satisfied.

The full support assumption $x_i(0)>0$ for all $i\in N$
does not imply a priori that any Nash equilibrium
$x_i(t), i\in N$ for some $t\in\T(c), c\in \R_+$ also has the full support property.
In the following, however,  we show that given the full support property
for $t=0$, it continues to hold for all $t\in\T(c), c\in \R_+$.

\begin{lemma}\label{lem:pos}
Let $c>0$ and $t\in\T(c)$. If $x_i(0)>0$ for all $i\in N$, then $x_i(t)>0$ and $t_i>0$ for all $i\in N$.
\end{lemma}
\begin{proof}
By Lemma \ref{lem:mon}, if $x_i(0)>0$ for all $i\in N$, then $x_i(0,t_{-i})>0$ for all $i\in N$ and all $t_{-i}$. As the profit function is continuous in the toll vector (Hayrapetyan et al.~\cite{HaTaWe07}), this implies that for all $t_{-i}$, there is some $t_i$ so that $\Pi_i(t_i,t_{-i})>0$. Hence if $t\in\mathcal{T}(c)$, then for all $i\in N$, $\Pi_i(t_i,t_{-i})>0$ and thus $x_i(t)>0$ and $t_i>0$.
\end{proof}

In the following lemma we derive an explicit price representation of Nash equilibria.
\begin{lemma}\label{lem:pri}
Let $c\in\mathbb{R}_+\cup\{\infty\}$ and $t\in\mathcal{T}(c)$. If $x_i(t)>0$ for all $i\in N$, then for all $i\in N$, 
$$t_i=\min\left\{\left(\ell'_i(x_i(t))+\frac{1}{\sum_{j\neq i}\frac{1}{\ell'_j(x_j(t))}}\right)\cdot x_i(t),c\right\}.$$
\end{lemma}
\begin{proof}
Assume that $x_i(t)> 0$ for all $i\in N$. Given $t_{-i}\in\mathbb{R}^{N\setminus\{i\}}_+$, each firm $i\in N$ solves the following maximization problem:
\begin{align*}
\max\hskip6pt &t_i\cdot x_i\\
s.t.\hskip6pt & \ell_j(x_j)+t_j= K\text{ for all }j\in N,\\
& \sum_{j\in N}x_j=1,\\
& t_i\leq c.
\end{align*}
The corresponding Lagrangian is
\begin{align*}
&\mathcal{L}(t_i,x,K,\lambda,\mu,\nu)\\
&=t_i\cdot x_i-\sum_{j\in N}\lambda_j\cdot(\ell_j(x_j)+t_j-K)-\nu\cdot(\sum_{j\in N}x_j-1)-\mu\cdot (t_i-c).
\end{align*}
So the Kuhn-Tucker conditions reduce to
\begin{align}
\frac{\partial\mathcal{L}(t_i,x,K,\lambda,\mu,\nu)}{\partial t_i}&=x_i-\lambda_i-\mu=0,\label{eq1}\\
\frac{\partial\mathcal{L}(t_i,x,K,\lambda,\mu,\nu)}{\partial x_i}&=t_i-\lambda_i\cdot \ell'_i(x_i)-\nu=0,\label{eq2}\\
\frac{\partial\mathcal{L}(t_i,x,K,\lambda,\mu,\nu)}{\partial x_j}&=-\lambda_j\cdot \ell'_j(x_j)-\nu=0,\label{eq3}\\
\frac{\partial\mathcal{L}(t_i,x,K,\lambda,\mu,\nu)}{\partial K}&=\sum_{j\in N}\lambda_j=0.\label{eq4}
\end{align}

By \eqref{eq1},
\begin{align}
\lambda_i=x_i-\mu
.\label{eq8}
\end{align}

By \eqref{eq2} and plugging in \eqref{eq8},
\begin{align}
t_i=\ell'_i(x_i)\cdot \lambda_i+\nu=\ell'_i(x_i)\cdot (x_i-\mu)+\nu.\label{eq5}
\end{align}

By \eqref{eq3}, for all $j\neq i$,
\begin{align}\lambda_j=-\nu\cdot\frac{1}{\ell'_j(x_j)}.\label{eq7}
\end{align}

By \eqref{eq7},
\begin{align}
\sum_{j\neq i}\lambda_j=-\nu\cdot\sum_{j\neq i}\frac{1}{\ell'_j(x_j)}\label{eq9}
\end{align}

By \eqref{eq4}, \eqref{eq8} and \eqref{eq9},
\begin{align}
\nu=\frac{\lambda_i-\sum_{j\in N}\lambda_j}{\sum_{j\neq i}\frac{1}{\ell'_j(x_j)}}=\frac{\lambda_i}{\sum_{j\neq i}\frac{1}{\ell'_j(x_j)}}=\frac{x_i-\mu}{\sum_{j\neq i}\frac{1}{\ell'_j(x_j)}}.\label{eq6}
\end{align}

Combining \eqref{eq5} and \eqref{eq6} yields
$$t_i=\left(\ell'_i(x_i)+\frac{1}{\sum_{j\neq i}\frac{1}{\ell'_j(x_j)}}\right)\cdot (x_i-\mu).$$

If $t_i<c$, then we know $\mu=0$, and thus,
$$t_i=\left(\ell'_i(x_i)+\frac{1}{\sum_{j\neq i}\frac{1}{\ell'_j(x_j)}}\right)\cdot x_i.$$
\end{proof}

Now we derive a complete characterization of equilibria.

\begin{theorem}\label{lem:cha}
Let $x_i(0)>0$ for all $i\in N$ and let $c\in\R_+$.
The tuple $(t,x)$ is a $c$-capped Nash equilibrium if and only
if the following conditions hold for some $K>0$:
\begin{align}
a_i x_i+b_i+t_i&=K \text{ for all }i\in N,\label{char1}\\
\sum_{i\in N}x_i&=1,\label{char2}\\
t_i&=\min\left\{\left(a_i+\frac{1}{\sum_{j\neq i}\frac{1}{a_j}}\right)\cdot x_i,c\right\} \text{ for all }i\in N,
\label{char3}\\
x_i&> 0 \text{ for all }i\in N.\label{char5}
\end{align}
\end{theorem}
\begin{proof}
We first show $\Rightarrow$: Conditions~\eqref{char1} and \eqref{char5} follow from Lemma~\ref{lem:pos} and the Wardrop condition. Condition~\eqref{char2} is trivial.
Condition~\eqref{char3} follows from Lemma~\ref{lem:pri}.

Now we prove $\Leftarrow$: let $(t',x')$ be a tuple that satisfies \eqref{char1}-\eqref{char5}. We want to show that $(t',x') =(t,x)$, where $(t,x)$ is a Nash equilibrium, which by Proposition \ref{pro:exi} exists. By \eqref{char3}, we get that tolls $t'$ are feasible w.r.t. $c$. Condition \eqref{char1} implies that $x'$ is a Wardrop equilibrium with full support with respect to $t'$. Hence,
similarly as in \eqref{eq:monotonicity}, we get\begin{equation}\label{eq:monotonicity-}
\sum_{j\in N}(t_j-t'_j)\cdot(x_j-x'_j)\leq 0.
\end{equation}

Assume by contradiction that there is $i\in N$ with $t'_i<t_i$ (the case $t'_i>t_i$ follows similarly). By \eqref{char3}, we get
\[ t'_i=\left(a_i+\frac{1}{\sum_{j\neq i}\frac{1}{a_j}}\right)\cdot x'_i<t_i\leq\left(a_i+\frac{1}{\sum_{j\neq i}\frac{1}{a_j}}\right)\cdot x_i.\]
From this $x'_i<x_i$ follows.
Putting things together, we get
\[ \sum_{j\in N}(t_j-t'_j)(x_j-x'_j)>0\]
a contradiction to~\eqref{eq:monotonicity-}.
\end{proof}

\begin{corollary}\label{lem:uni}
There is at most one tuple $(t,x)$ that satisfies \eqref{char1}-\eqref{char5}, thus the Nash equilibrium is unique.
\end{corollary}

If $\T(c)$ is singular, we let $t(c)$ denote $t\in\T(c)$, and $x(c)$ denote $x(t(c))$.

\section{Optimal price caps for affine latencies}
Assume that $\ell_i(x_i)=a_i\cdot x_i+b_i$ with $a_i> 0$ and $b_i\geq 0$, and $x_i(0)>0$ for all $i\in N$.
Now we have everything together to derive an optimal polynomial
time algorithm for networks with affine latencies, which is the main result of this section.

\IncMargin{1em}
\begin{algorithm}
\SetKwData{Left}{left}\SetKwData{This}{this}\SetKwData{Up}{up}
\SetKwFunction{Union}{Union}\SetKwFunction{FindCompress}{FindCompress}
\SetKwInOut{Input}{input}\SetKwInOut{Output}{output}
\Input{$I=(N,(\ell_i)_{i\in N})$}
\Output{An optimal price cap $c^*$}
\BlankLine
\textbf{initialize}\\
$c_0\leftarrow\infty$\;
$A(c_0)\leftarrow\emptyset$\;
$j\leftarrow 0$\;
\While{$A(c_j)\neq N$}{
  $K(c)\leftarrow \frac{1+\sum\limits_{k\in A(c_j)}\frac{b_k+c}{a_k}+\sum\limits_{k\in N\setminus A(c_j)}\frac{b_k}{2a_k+\frac{1}{\sum\limits_{l\in N\setminus\{k\}}1/a_l}}}{\sum\limits_{k\in A(c_j)}\frac{1}{a_k}+\sum\limits_{k\in N\setminus A(c_j)}\frac{1}{2a_k+\frac{1}{\sum\limits_{l\in N\setminus\{k\}}1/a_l}}}$\;
   $x_i(c)\leftarrow \frac{K(c)-b_i-c}{a_i}$ for all $i\in A(c_j)$\; \label{equ1}
  $x_i(c)\leftarrow \frac{K(c)-b_i}{2a_i+\frac{1}{\sum_{j\neq i}1/a_j}}$ for all $i\in N\setminus A(c_j)$\; \label{equ2}
 $ t_i(c)\leftarrow c$ for all $i\in A(c_j)$\; \label{eq:toll1}
$t_i(c)\leftarrow\left(a_i+\frac{1}{\sum_{l\neq i}1/a_l}\right)\cdot x_i(c)$ for all $i\in N\setminus A(c_j)$\; \label{eq:toll2}\vspace{0.15cm}
  $c_{j+1}^i \leftarrow \max\left\{c\middle\vert \left(a_i+\frac{1}{\sum_{l\neq i}1/a_l}\right)\cdot x_i(c)= c\right\} \text{ for all } i\in N\setminus A(c_j)$\;\vspace{0.15cm}\label{eqc1}
  $c_{j+1} \leftarrow\max\left\{c_{j+1}^i\middle\vert i\in N\setminus A(c_j)\right\}$ \;\label{eqc2}
  $A(c_{j+1})\leftarrow A(c_j) \cup \arg\max\left\{c_{j+1}^i\mid i\in N\setminus A(c_j)\right\}$\;
  $j\leftarrow j+1$\;
  $c^*_j\leftarrow \arg\min_{c\in [c_j,c_{j-1}]} \sum_{i\in N}\ell_i(x_i(c))\cdot x_i(c)$
}
\textbf{output  $c^*\in \arg\min\left\{C(x(c^*_k)), k=1,\dots, j\right\}$}
\caption{An optimal algorithm for affine latencies.}\label{alg:opt}
\end{algorithm}\DecMargin{1em}

\begin{theorem}
Let $\ell_i\in\mathcal{L}_1$ and $x_i(0)>0$ for all $i\in N$. Algorithm~\ref{alg:opt} computes in polynomial time an optimal price cap.
\end{theorem}
\begin{proof}
We show by induction on the iterations 
of the while loop of Algorithm~\ref{alg:opt} (indexed by $k\in\N$) that the algorithm
computes breakpoints $c_1>\dots > c_k>\dots > c_j$ 
so that in the intervals
$\mathcal{I}_1=[c_1,\infty)$ and $\mathcal{I}_k=[c_k,c_{k-1}]$ for $k=2,\ldots,j+1$  with $c_{j+1}=0$ the following invariant holds: for all  $k=1,\dots,j+1$, the flows
$x_i(c), i\in N$ as defined in Lines~\ref{equ1} and~\ref{equ2}
together with prices $t_i(c), i\in N$ as defined in Lines~\ref{eq:toll1} and~\ref{eq:toll2}
constitute the unique Nash equilibrium for all $c\in\mathcal{I}_k$.

Consider the base case $k=1$.  
First, observe that Line~\ref{eqc2} is well defined
since
the maximum in  Line~\ref{eqc2} obviously exists (it is attained at $c_1=\max_{i\in N}t_i(\infty)$).
The parameterized flow $x(c)$ as defined in lines Line~\ref{equ1} and~\ref{equ2} is a solution of the following system of linear equations:
\begin{align*}
a_i\cdot x_i(c)+b_i+\left(a_i+\frac{1}{\sum_{l\neq i}1/a_l}\right)\cdot x_i(c)&=K(c)\text{ for all }i\in N,\\
\sum_{i\in N}x_i(c)&=1.
\end{align*}
By Theorem \ref{lem:cha} and Corollary \ref{lem:uni}, $x(c),t(c)$ is the unique Nash equilibrium for all $c\in \mathcal{I}_1$.

For the inductive step $k\rightarrow k+1$, assume $x(c), t(c)$ as defined in Lines~\ref{equ1}-\ref{eq:toll2} is the unique Nash equilibrium for all $c\in\mathcal{I}_\ell, \ell=1,\ldots,k$.

First, we show again that the maximum $c_{k+1}^i, i\in N\setminus A(c_k)$ in Line~\ref{eqc1} exists so that Lines~\ref{eqc1} and~\ref{eqc2} are well defined, and, 
 that $c_{k+1}=\max\{c_{k+1}^i \vert i\in N\setminus A(c_k)\}<c_k$ as claimed. 
To see
this, observe that for $i\in N\setminus A(c_k)$,
we have $\left(a_i+\frac{1}{\sum_{l\neq i}1/a_l}\right)\cdot x_i(c_k)<c_k$.
On the other hand, by the assumption that $x_i(0)>0$ for all $i\in N$ we have
$ \left(a_i+\frac{1}{\sum_{l\neq i}1/a_l}\right)\cdot x_i(0)>0.$
As the function $\left(a_i+\frac{1}{\sum_{l\neq i}1/a_l}\right)\cdot x_i(c)$ is continuous in $c$,  by the intermediate value theorem, there
exists $\bar c_{k+1}^i$ with 
\[ \left(a_i+\frac{1}{\sum_{l\neq i}1/a_l}\right)\cdot x_i(\bar c_{k+1}^i)=\bar c_{k+1}^i,\]
 implying $c_{k+1}^i<c_k$ and thus $c_{k+1}<c_k$.

We next prove that $x(c),t(c)$  as defined in Lines~\ref{equ1}-\ref{eq:toll2}
 is the unique Nash equilibrium for all $c\in\mathcal{I}_{k+1}$. We prove this by showing that $x(c), t(c)$ satisfies 
 the conditions of Theorem~\ref{lem:cha}. Observe that $x(c)$ is a solution of the following system of linear equations:

\begin{align*}
a_i\cdot x_i(c)+b_i+t_i&=K(c)\text{ for all }i\in N,\\
t_i&=c\text{ for all }i\in A(c_k),\\
t_i&=\left(a_i+\frac{1}{\sum_{l\neq i}1/a_l}\right)\cdot x_i(c)\text{ for all }i\in N\setminus A(c_k),\\
\sum_{i\in N}x_i(c)&=1.
\end{align*}
By Theorem \ref{lem:cha}, it is suffices to show that $t_i(c)$ satisfies condition~\eqref{char3} for all $i\in N, c\in \mathcal{I}_{k+1}$. From Lines~\ref{eqc1} and \ref{eqc2} we get  $t_i(c)\leq c$ for all $i\in N\setminus A(c_k)$ and all $c\in \mathcal{I}_{k+1}$. It remains to show that for all $i\in A(c_k)$ and all $c\in \mathcal{I}_{k+1}$, $$\left(a_i+\frac{1}{\sum_{l\neq i}1/a_l}\right)\cdot x_i(c)\geq c.$$
We know by the induction hypothesis that $x(c_k)$ is the unique Nash equilibrium, and thus satisfies
 $\left(a_i+\frac{1}{\sum_{l\neq i}1/a_l}\right)\cdot x_i(c_k)\geq c_k$ for all $i\in A(c_k)$. For all $c\in(c_{k+1},c_k)$, we have
$$\frac{\partial K(c)}{\partial c}=\frac{\sum\limits_{i\in A(c_j)}\frac{1}{a_i}}{\sum\limits_{i\in A(c_j)}\frac{1}{a_i}+\sum\limits_{i\in N\setminus A(c_j)}\frac{1}{2a_i+\frac{1}{\sum\limits_{l\in N\setminus\{i\}}1/a_l}}}\leq 1,$$
and thus by Line \ref{equ1}, we get for all $i\in A(c_k)$ and all $c\in(c_{k+1},c_k)$, $$\frac{\partial x_i(c)}{\partial c}\leq 0.$$ So when decreasing $c$ from $c_k$ to $c_{k+1}$, it follows that $\left(a_i+\frac{1}{\sum_{l\neq i}1/a_l}\right)\cdot x_i(c)$ increases, and thus $\left(a_i+\frac{1}{\sum_{l\neq i}1/a_l}\right)\cdot x_i(c)\geq c$ for all $i\in A(c_k)$ and all $c\in\mathcal{I}_{k+1}$, which completes the induction proof.

For each $\mathcal{I}_k$, with $k=1,\ldots,j+1$, the objective function $$C(x(c))=\sum_{i\in N}\ell_i(x_i(c))\cdot x_i(c)$$ is quadratic in $c$, thus, we can find a local minimum by comparing the two endpoints and a possible interior minimum point.
For the interior minimum, we just need to check first-order optimality conditions, thus,  solving a linear equation in $c$.
This way, we have effectively partitioned the search space $\R_+$
into at most $n$ intervals. As we solved each segment
optimally, taking the best solution leads to the optimal $c^*$.
\end{proof}

\begin{remark}
The algorithm is still polynomial for arbitrary demands $d\geq 0$
since solving the linear equation systems appearing in the algorithm can be done in polynomial
time in the encoding length of $d$.
\end{remark}

The following example demonstrates the calculation process of the algorithm.

\begin{example}
Consider the network of Figure \ref{fig:alg}.
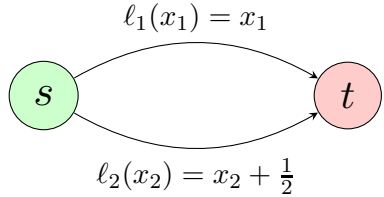
\begin{figure}[ht]
\centering
\begin{tikzpicture}
\node[draw,circle,scale=1.5,fill=green!20] (1) at (0,0) {$s$};
\node[draw,circle,scale=1.4,fill=red!20] (2) at (4,0) {$t$};
\draw[-stealth] (1)[out=30,in=150] to node[above]{$\ell_1(x_1)=x_1$} (2);
\draw[-stealth] (1)[out=-30,in=-150] to node[below]{$\ell_2(x_2)=x_2+\frac{1}{2}$} (2);
\end{tikzpicture}
\caption{A demonstration of the algorithm.}
\label{fig:alg}
\end{figure}

By Lemma \ref{lem:pri}, the uncapped Nash equilibrium prices can be found by solving the following system of linear equations
\begin{align*}
x_1(c)+2x_1(c)&=K(c),\\
x_2(c)+\frac{1}{2}+2x_2(c)&=K(c),\\
x_1(c)+x_2(c)&=1,
\end{align*}
and thus are given by $t(\infty)=\left(\frac{7}{6},\frac{5}{6}\right)$. Initialize $c_1=\frac{7}{6}$ and $A(c_1)=\{1\}$. Solve the following system of linear equations
\begin{align*}
x_1(c)+c&=K(c),\\
x_2(c)+\frac{1}{2}+2x_2(c)&=K(c),\\
x_1(c)+x_2(c)&=1,
\end{align*}
yields $x_2(c)=\frac{2c+1}{8}$. Solving $2\cdot\frac{2c+1}{8}=c$ yields $c=\frac{1}{2}$. Hence, $c_2=\frac{1}{2}$ and $A(c_2)=\{1,2\}$.

Therefore,
\begin{align*}
C(x(c))=\begin{cases}
\frac{3}{4}, &\mbox{ if }c\leq\frac{1}{2},\\
\frac{4c^2-9c+27}{32}, &\mbox{ if }\frac{1}{2}<c\leq\frac{7}{6},\\
\frac{13}{18}, &\mbox{ if }c>\frac{171}{224},
\end{cases}
\end{align*}
which is minimized for $c=1$. Notice that we induce the optimal flow with a uniform price cap of $c=1$.
\end{example}

The last example of this section shows that the full Wardrop support assumption is important. Without this assumption, a Nash equilibrium need not be unique.

\begin{example}\label{ex:mul}
Consider the network of Figure \ref{fig:mul}.
\begin{figure}[ht]
\centering
\begin{tikzpicture}
\node[draw,circle,scale=1.5,fill=green!20] (1) at (0,0) {$s$};
\node[draw,circle,scale=1.4,fill=red!20] (2) at (4.5,0) {$t$};
\draw[-stealth] (1)[out=45,in=135] to node[above]{$\ell_1(x_1)=x_1$} (2);
\draw[-stealth] (1)[out=0,in=180] to node[above]{$\ell_2(x_2)=x_2$} (2);
\draw[-stealth] (1)[out=-45,in=-135] to node[below]{$\ell_3(x_3)=\frac{x_3}{2}+\frac{6}{5}$} (2);
\end{tikzpicture}
\caption{Multiple Nash equilibria.}
\label{fig:mul}
\end{figure}
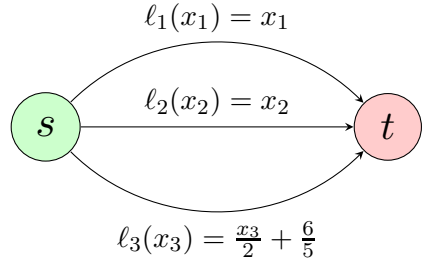
Assume that $t_3=0$. Then
\begin{align*}
B^{\infty}_1(t_2,t_3)=\begin{cases}
t_1=\frac{t_2+1}{2}, &\mbox{ if }0\leq t_2\leq\frac{3}{5},\\
t_1=\frac{7-5t_2}{5}, &\mbox{ if }\frac{3}{5}<t_2\leq\frac{5}{7},\\
t_1=\frac{5t_2+17}{30}, &\mbox{ if }\frac{5}{7}<t_2\leq\frac{17}{10},\\
t_1=\frac{17}{20}, &\mbox{ if }t_2>\frac{17}{10}.
\end{cases}
\end{align*}
and
\begin{align*}
B^{\infty}_2(t_1,t_3)=\begin{cases}
t_2=\frac{t_1+1}{2}, &\mbox{ if }0\leq t_1\leq\frac{3}{5},\\
t_2=\frac{7-5t_1}{5}, &\mbox{ if }\frac{3}{5}<t_1\leq\frac{5}{7},\\
t_2=\frac{5t_1+17}{30}, &\mbox{ if }\frac{5}{7}<t_1\leq\frac{17}{10},\\
t_2=\frac{17}{20}, &\mbox{ if }t_1>\frac{17}{10}.
\end{cases}
\end{align*}
\begin{figure}[ht]
\centering
\begin{tikzpicture}[scale=1]
\draw[very thin,color=gray] (-0.1,-0.1) grid (4.1,4.1);
\node[below left] at (0,0){0};
\node[below] at (2,0){$\frac{1}{2}$};
\node[below] at (4,-0.1){1};
\node[left] at (0,2){$\frac{1}{2}$};
\node[left] at (0,4){1};
\draw[->] (-0.2,0) -- (4.2,0) node[right] {$t_1$};
\draw[->] (0,-0.2) -- (0,4.2) node[above] {$t_2$};
%\draw[dashed,very thin,color=gray] (2.4,0) -- (2.4,3.2);
%\draw[dashed,very thin,color=gray] (0,3.2) -- (2.4,3.2);
%\draw[dashed,very thin,color=gray] (3.2,0) -- (3.2,2.4);
%\draw[dashed,very thin,color=gray] (0,2.4) -- (3.2,2.4);
\draw[color=blue] (2,0) -- (3.2,2.4);
\draw[color=blue] (3.2,2.4) -- (96/35,20/7);
\draw[color=blue] (96/35,20/7) -- (44/15,4);
\draw[color=red,dashed] (0,2) -- (2.4,3.2);
\draw[color=red,dashed] (2.4,3.2) -- (20/7,96/35);
\draw[color=red,dashed] (20/7,96/35) -- (4,44/15);
\end{tikzpicture}
\caption{Best reply correspondence: in blue for player 1, and in (dashed) red for player 2.}
\label{fig:br1}
\end{figure}
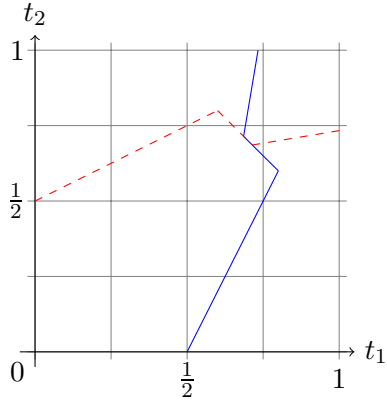
Combining $B^{\infty}_1(t_2,t_3)$ and $B^{\infty}_2(t_1,t_3)$ (see Figure \ref{fig:br1}) implies that the set

\noindent $\left\{\left(t_1,\frac{7}{5}-t_1,0\right)\mid\frac{24}{35}\leq t_1\leq\frac{5}{7}\right\}$ are Nash equilibria.

\end{example}
\begin{remark} The network pricing game of Example \ref{ex:mul} exhibits another interesting phenomenon: if $\ell_3(x_3)=a_3\cdot x_3+\frac{6}{5}$ with $a_3\leq 2/3$, then every equilibrium yields a profit of 0 for firm 3, whereas if $\ell_3(x_3)=a_3\cdot x_3+\frac{6}{5}$ with $a_3> 2/3$, then the equilibrium is unique and yields a strictly positive profit for firm 3. This implies that firm 3 is worse off by being congestion free than by being sufficiently congestion dependent. The reason is that in the former case firm 3 is too competitive, and thus its two competitors set prices in such a way that prevents the third firm from entering the market.
\end{remark}

\section{Optimal price caps for duopolies}
In this section, we compare the latency costs of optimal price caps
with those of optimal solutions minimizing total congestion (ignoring price competition).  The aim of this comparison is to  reveal
the possible strength of introducing price caps as a mechanism.
Our results, however, only hold for duopolies, that is, we assume $n=2$
for the following analysis. We conjecture though that our results carry over
to general $n$  (in the spirit of Pigou instances).

Let $c^*$ denote an optimal price cap (a solution to problem \ref{eq:problem}), let $x^*$ denote the optimal flow (ignoring price caps), and let $\L\subseteq\L_c$ be a class of latency functions. We are interested in the ratio $\rho(\L)$ between the latency costs at the optimal price cap and the latency costs of the optimal flow, defined by, for all $\ell\in\L$,
\[ C(x(c^*))\leq \rho(\L)\cdot C(x^*). \]

Our main results are as follows.
\begin{theorem}\label{thm:main1}
Let $n=2$ and $\T(c)\neq\emptyset$ for all $c\in\R_+$. Then 
\begin{enumerate}
\item[(i)] $\rho(\L_d)\leq(1-\frac{d}{2(d+1)^{(d+1)/d}})^{-1}$,
\item[(ii)] $\rho(\L_c)\leq 2$.
\end{enumerate}
\end{theorem}

Suppose that we have a Nash equilibrium in which one firm receives all the flow, say firm $1$. The proof of Lemma \ref{lem:pri} implies that in that case $\ell_1(1)+\ell'_1(1)\leq\ell_2(0)$. This again implies that the duopoly and optimal flow coincide. So in the remainder of this section, we assume that $x_i(t)>0$ for all $i=1,2$ for $t\in \T(\infty)$.

Before we prove the main result of this section, we introduce some helpful lemmas.
\begin{lemma}\label{lem:2eq}
Let $t\in \T(\infty)$. Suppose that $x_i(t)>0$ for $i=1,2$. Then $x_1(0)\geq x_1(t)$ if and only if $x_1(t)\geq\frac{1}{2}$.
\end{lemma}
\begin{proof}
Suppose that $x_1(0)\geq x_1(t)$. Then
$$\ell_1(x_1(t))\leq\ell_1(x_1(0))\leq\ell_2(x_2(0))\leq\ell_2(x_2(t)),$$
where the first inequality follows from $x_1(t)\leq x_1(0)$ and $\ell_1$ increasing, the second inequality from $x_1(0)\geq x_1(t)>0$, and the third inequality from $x_2(0)\leq x_2(t)$ and $\ell_2$ increasing. Since $x_i(t)>0$ for $i=1,2$, we have that $\ell_1(x_1(t))+t_1=\ell_2(x_2(t))+t_2$ and thus $t_1\geq t_2$. Hence Lemma \ref{lem:pri} implies that $x_1(t)\geq x_2(t)$ and thus $x_1(0)\geq x_1(t)\geq\frac{1}{2}$.

Suppose that $x_1(0)<x_1(t)$. Then
$$\ell_2(x_2(t))<\ell_2(x_2(0))\leq\ell_1(x_1(0))<\ell_1(x_1(t)),$$
where the first inequality follows from $x_2(t)<x_2(0)$ and $\ell_2$ increasing, the second inequality from $x_2(0)>0$, and the third inequality from $x_1(0)<x_1(t)$ and $\ell_1$ increasing. Since $x_i(t)>0$ for $i=1,2$, we have that $\ell_1(x_1(t))+t_1=\ell_2(x_2(t))+t_2$ and thus $t_1< t_2$. Hence Lemma \ref{lem:pri} implies that $x_1(t)<x_2(t)$ and thus $x_1(t)<\frac{1}{2}$.
\end{proof}

Define 
$$\mu_1(\ell_i)=\sup_{x_i,x^*_i\geq0}\frac{(\ell_i(x_i)-\ell_i(x^*_i))\cdot x^*_i}{\ell_i(x_i)\cdot x_i}$$
for each $\ell_i\in\mathcal{L}$ and
$$\mu_1(\mathcal{L})=\sup_{\ell_i\in\mathcal{L}}\mu_1(\ell_i).$$
The parameter $\mu_1(\mathcal{L})$ is a measure of the steepness of the class of allowable latency functions that is well studied in the context of bounding the price of anarchy in routing games. It is well known that the price of anarchy in routing games can be bounded as a function of the class of allowable latency functions, but usually not in terms of other characteristics of the instance like the topology of the network. For more details, see, for example, Correa et al.~\cite{CorSS04} and Roughgarden~\cite{Ro03}. Observe that $\mu_1(\mathcal{L})\in[0,1]$.
\begin{lemma}\label{lem:we}
If $x_i(0)>x^*_i$ and $x_i(0)\leq\frac{1}{2}$ for some $i=1,2$, then $C(x(0))\leq\frac{1}{1-\mu_1(\mathcal{L})/2}\cdot C(x^*)$.
\end{lemma}
\begin{proof}
W.l.o.g. suppose that $x_1(0)>x^*_1$ and $x_1(0)\leq\frac{1}{2}$. By Lemma \ref{lem:vi}, we have
\begin{align*}
C(x(0))&\leq C(x^*)+\sum_{i=1}^2(\ell_i(x_i(0))-\ell_i(x^*_i))\cdot x^*_i,\\
&\leq C(x^*)+(\ell_1(x_1(0))-\ell_1(x^*_1))\cdot x^*_1,
\end{align*}
where the second inequality follows from $x_2(0)< x^*_2$ and $\ell_2$ increasing. By definition of $\mu_1(\mathcal{L})$,
$$\frac{(\ell_1(x_1(0))-\ell_1(x^*_1))\cdot x^*_1}{\ell_1(x_1(0))\cdot x_1(0)}\leq\mu_1(\mathcal{L}).$$
The lemma then follows because $\ell_1(x_1(0))\leq C(x(0))$ and $x_1(0)\leq\frac{1}{2}$.
\end{proof}

Define 
$$\mu_2(\ell_i)=\sup_{x_i\geq\frac{1}{2},0\leq x^*_i\leq x_i}\frac{(\ell_i(x_i)-\ell_i(x^*_i))\cdot(x^*_i+1 -2x_i)}{\ell_i(x_i)}$$
for each $\ell_i\in\mathcal{L}$ and
$$\mu_2(\mathcal{L})=\sup_{\ell_i\in\mathcal{L}}\mu(\ell_i).$$
The parameter $\mu_2(\mathcal{L})$ is a new smoothness parameter that takes into account the pricing behavior of the firms. Observe that $\mu_2(\mathcal{L})\in[0,1]$.
\begin{lemma}\label{lem:oe}
Let $t\in \T(\infty)$ and $x_i(t)>0$ for all $i=1,2$. If $x_i(t)>x^*_i$ and $x_i(t)\geq\frac{1}{2}$ for some $i=1,2$, then $C(x(t))\leq\frac{1}{1-\mu_2(\mathcal{L})}\cdot C(x^*)$.
\end{lemma}
\begin{proof}
W.l.o.g. suppose that $x_1(t)>x^*_1$ and $x_1(t)\geq\frac{1}{2}$. By Lemma \ref{lem:vi} and \ref{lem:pri}, we have
\begin{align*}
&C(x(t))\\
&\leq C(x^*)+\sum_{i=1}^2(\ell_i(x_i(t))-\ell_i(x^*_i))\cdot x^*_i+\sum_{i=1}^2 t_i(x^*_i-x_i(t)),\\
&\leq C(x^*)+\sum_{i=1}^2(\ell_i(x_i(t))-\ell_i(x^*_i))\cdot x^*_i\\
&-(\ell_1' (x_1(t))+\ell_2' (x_2(t)))\cdot(x_1(t)-x^*_1)\cdot (2x_1(t)-1),\\
&\leq C(x^*)+(\ell_1(x_1(t))-\ell_1(x^*_1))\cdot x^*_1-\ell_1' (x_1(t))\cdot(x_1(t)-x^*_1)\cdot (2x_1(t)-1),\\
&\leq C(x^*)+(\ell_1(x_1(t))-\ell_1(x^*_1))\cdot (x^*_1+1-2x_1(t)).
\end{align*}
where the second inequality follows from Lemma \ref{lem:pri}, $x_2(t)\leq x^*_2$ and $\ell_2$ increasing, the third from $\ell_2'(x_2(t))\geq0$ and the fourth from convexity of $\ell_1$, $\ell_1' (x_1(t))\cdot(x_1(t)-x^*_1)\geq \ell_1(x_1(t))-\ell_1(x^*_1)$. Since
$$\frac{(\ell_1(x_1(t))-\ell_1(x^*_1))\cdot (x^*_1+1-2x_1(t))}{\ell_1(x_1(t))}\leq\mu_2(\mathcal{L}_c),$$
and, by Lemma \ref{lem:pri}, $t_1\geq t_2$, we have $\ell_1(x_1(t))\leq C(x(t))$ and the lemma follows.
\end{proof}

\begin{proof}[Proof of Theorem \ref{thm:main1}]
Let $t\in \T(\infty)$ and $x_i(t)>0$ for all $i=1,2$. Suppose that either $x_1(t)\leq x^*_1<x_1(0)$, or $x_1(0)<x^*_1<x_1(t)$. By the Theorem of the Maximum (Berge~\cite{Be63}), the profit function of each firm is continuous in $c$. Given that the profit function of each firm is continuous in the toll vector (Hayrapetyan et al.~\cite{HaTaWe07}), and the toll vector continuously changes the induced flow (Beckmann et al.~\cite{BeMcWi56}), there exists a price cap $c$ such that $x_1(c)=x^*_1$. So, we can assume that $x^*_1<x_1(0)$ and $x^*_1<x_1(t)$.

Suppose that $x_1(t)\leq\frac{1}{2}$. Then by Lemma \ref{lem:2eq}, $x_1(0)\leq x_1(t)\leq\frac{1}{2}$. Since $\mu_1(\L_d)\leq \frac{d}{(d+1)^{(d+1)/d}}$ (Correa et al.~\cite{CorSS04}) and $\mu_1(\L_c)\leq 1$, the result follows by Lemma \ref{lem:we}.

Suppose that $x_1(t)\geq\frac{1}{2}$. Then
\begin{align*}
\mu_2(\L_d)&=\sup_{x_i(t)\geq\frac{1}{2},0\leq x^*_i\leq x_i(t)}\frac{\left(\ell_1(x_1(t))-\ell_1(x^*_1)\right)\cdot(x^*_1+1-2x_1(t))}{\ell_1(x_1(t))}\\
&\leq\frac{d}{2(d+1)^{(d+1)/d}},
\end{align*}
and
\begin{align*}
\mu_2(\L_c)&=\sup_{x_i(t)\geq\frac{1}{2},0\leq x^*_i\leq x_i(t)}\frac{\left(\ell_1(x_1(t))-\ell_1(x^*_1)\right)\cdot(x^*_1+1-2x_1(t))}{\ell_1(x_1(t))}\\
&\leq\sup_{\frac{1}{2}\leq x_1(t)} 1-x_1(t)\leq\frac{1}{2},
\end{align*}
and the result follows by Lemma \ref{lem:oe}.
\end{proof}

The following example shows that the bound of $8/7$ in Theorem \ref{thm:main1} for affine latencies is tight for duopolies, and is a lower bound for arbitrary parallel graphs.

\begin{example}
Consider the network of Figure \ref{fig:aff} with $n\geq2$.
\begin{figure}[ht]
\centering
\begin{tikzpicture}
\node[draw,circle,scale=1.5,fill=green!20] (1) at (0,0) {$s$};
\node[draw,circle,scale=1.4,fill=red!20] (2) at (4,0) {$t$};
\node at (2,1/3) {$\vdots$};
\draw[-stealth] (1)[out=30,in=150] to node[above]{$\ell_1(x_1)=x_1$} (2);
\draw[-stealth] (1)[out=-30,in=-150] to node[above=5pt,scale=0.8]{$\ell_{n-1}(x_{n-1})=x_{n-1}$} (2);
\draw[-stealth] (1)[out=-45,in=-135] to node[below]{$\ell_n(x_n)=\frac{1}{2(n-1)}$} (2);
\end{tikzpicture}
\caption{The bound for affine latencies is tight.}
\label{fig:aff}
\end{figure}
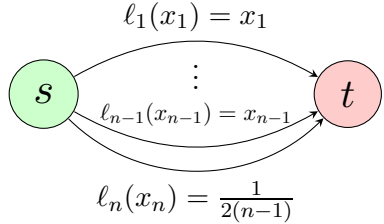

Let $t\in\T(\infty)$. By Lemma \ref{lem:pri}, $t=x(t)=x(0)=\left(\frac{1}{2(n-1)},\ldots,\frac{1}{2(n-1)},\frac{1}{2}\right)$ and $x^*=\left(\frac{1}{4(n-1)},\ldots,\frac{1}{4(n-1)},\frac{3}{4}\right)$. 

If $0\leq c\leq\frac{1}{2(n-1)}$, then we have $t_1=\ldots=t_n=c$ and thus $x(c)=\left(\frac{1}{2(n-1)},\ldots,\frac{1}{2(n-1)},\frac{1}{2}\right)$. If $c>\frac{1}{2(n-1)}$, then $t_1=\ldots=t_2=\frac{1}{2(n-1)}$ and thus $x(c)=\left(\frac{1}{2(n-1)},\ldots,\frac{1}{2(n-1)},\frac{1}{2}\right)$. Hence any price cap is optimal. So $$\min_{c\in\R_+}\frac{C(x(c))}{C(x^*)}=\frac{8}{7}.$$
\end{example}

The main result in Theorem \ref{thm:main1} assumes that a $c$-capped Nash equilibrium exists for all $c\in\R_+$. The following example shows that an uncapped Nash equilibrium need not exist for quadratic latency functions.

\begin{example}\label{exa:non}
Consider the network of Figure \ref{fig:non}.
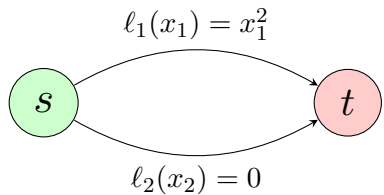
\begin{figure}[ht]
\centering
\begin{tikzpicture}
\node[draw,circle,scale=1.5,fill=green!20] (1) at (0,0) {$s$};
\node[draw,circle,scale=1.4,fill=red!20] (2) at (4,0) {$t$};
\draw[-stealth] (1)[out=30,in=150] to node[above]{$\ell_1(x_1)=x_1^2$} (2);
\draw[-stealth] (1)[out=-30,in=-150] to node[below]{$\ell_2(x_2)=0$} (2);
\end{tikzpicture}
\caption{No Nash equilibrium.}
\label{fig:non}
\end{figure}

Then
\begin{align*}
B^{\infty}_1(t_2)=\begin{cases}
t_1=\frac{2t_2}{3}, &\mbox{ if }0\leq t_2\leq 3,\\
t_1=t_2-1, &\mbox{ if }t_2>3,
\end{cases}
\end{align*}
and
\begin{align*}
B^{\infty}_2(t_1)=\begin{cases}
t_2=\argmax\limits_{t_2}\left(1-(t_2-t_1)^{1/2}\right)\cdot t_2, &\mbox{ if }0\leq t_1\leq\frac{1}{4},\\
t_2=t_1, &\mbox{ if }t_1\geq\frac{1}{4}.
\end{cases}
\end{align*}
Observe that $B^{\infty}_2(t_1)$ is not convex at $t_1=\frac{1}{4} $. Combining $B^{\infty}_1(t_2)$ and $B^{\infty}_2(t_1)$ (see Figure \ref{fig:br2}) implies that there is no uncapped Nash equilibrium. In fact, the result is even stronger: there is no $c$-capped Nash equilibrium whenever $c\geq\frac{1}{2}$.
\begin{figure}[ht]
\centering
\begin{tikzpicture}
\draw[very thin,color=gray] (-0.1,-0.1) grid (4.1,4.1);
\node[below left] at (0,0){0};
\node[below] at (2,0){$\frac{1}{4}$};
\node[below] at (4,0){$\frac{1}{2}$};
\node[left] at (0,2){$\frac{1}{4}$};
\node[left] at (0,4){$\frac{1}{2}$};
\draw[->] (-0.2,0) -- (4.2,0) node[right] {$t_1$};
\draw[->] (0,-0.2) -- (0,4.2) node[above] {$t_2$};
\draw[color=red,dashed] (2,2) -- (4,4);
\draw[color=red,dashed] (0,32/9) to[out=25,in=185] (2,4);
\draw[color=blue] (0,0) -- (8/3,4);
\end{tikzpicture}
\caption{Best reply correspondence: in blue for player 1, and in (dashed) red for player 2.}
\label{fig:br2}
\end{figure}
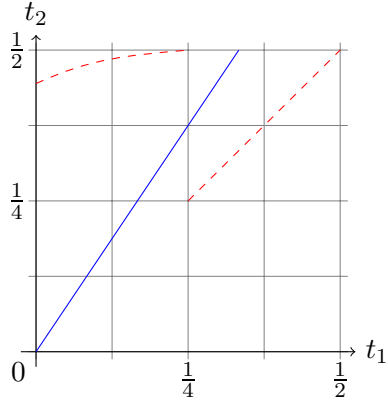
\end{example}
\begin{remark}
The nonexistence result in Example \ref{exa:non} violates the assumption that latency functions are strictly increasing. Nevertheless, we can change the latency function of link 2 into $\ell_2(x_2)=a_2\cdot x_2$ with $a_2<\sqrt{17}-4$ and obtain the same result: an uncapped Nash equilibrium need not exist.
\end{remark}

The next and final example shows that the latency costs at the optimal price cap can be arbitrarily worse than the optimal latency costs due to nonexistence of an uncapped Nash equilibrium.
\begin{example} \label{ex:bad}
Consider the network of Figure \ref{fig:poly}, where $d\geq3$.
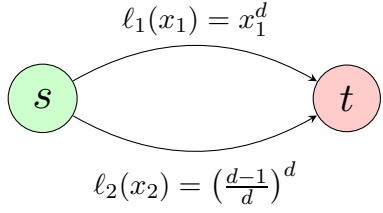
\begin{figure}[ht]
\centering
\begin{tikzpicture}
\node[draw,circle,scale=1.5,fill=green!20] (1) at (0,0) {$s$};
\node[draw,circle,scale=1.4,fill=red!20] (2) at (4,0) {$t$};
\draw[-stealth] (1)[out=30,in=150] to node[above]{$\ell_1(x_1)=x_1^d$} (2);
\draw[-stealth] (1)[out=-30,in=-150] to node[below]{$\ell_2(x_2)=\left(\frac{d-1}{d}\right)^d$} (2);
\end{tikzpicture}
\caption{No uncapped Nash equilibrium if $d\geq 3$.}
\label{fig:poly}
\end{figure}

Then
\begin{align*}
B^{\infty}_1(t_2)=\begin{cases}
t_1=\frac{d\cdot(t_2+\left(\frac{d-1}{d}\right)^d)}{d+1}, &\mbox{ if }0\leq t_2\leq d+1-\left(\frac{d-1}{d}\right)^d,\\
t_1=t_2+\left(\frac{d-1}{d}\right)^d-1, &\mbox{ if }t_2>d+1-\left(\frac{d-1}{d}\right)^d,
\end{cases}
\end{align*}
and
\begin{align*}
B^{\infty}_2(t_1)=\begin{cases}
t_2=\argmax\limits_{t_2}\left(1-(t_2-t_1+b_2)^{1/d}\right)\cdot t_2, &\mbox{ if }0\leq t_1\leq\frac{d\cdot(d-1)^{d-1}}{d^d},\\
t_2=t_1-\left(\frac{d-1}{d}\right)^d, &\mbox{ if }t_1\geq\frac{d\cdot(d-1)^{d-1}}{d^d}.
\end{cases}
\end{align*}
Combining $B^{\infty}_1(t_2)$ and $B^{\infty}_2(t_1)$ implies that there is only an equilibrium for all price caps $c\leq\left(\frac{d-1}{d}\right)^{d-1}$. If $c\leq\left(\frac{d-1}{d}\right)^{d-1}$, then $t_1(c)=t_2(c)=c$. Similarly as in Example \ref{exa:non}, an important reason for the nonexistence of equilibria seems to be the nonconvexity of the set of best replies for player 2 at $t_1=\left(\frac{d-1}{d}\right)^{d-1}$. Since $x(c)=x(0)=\left(\frac{d-1}{d},\frac{1}{d}\right)$ for all $c\leq\left(\frac{d-1}{d}\right)^{d-1}$, and $x^*=\left(\frac{d-1}{d\cdot(d+1)^{1/d}},1-\frac{d-1}{d\cdot(d+1)^{1/d}}\right)$, we have
$$\rho(\L_d)\geq\min_{0\leq c\leq \left(\frac{d-1}{d}\right)^{d-1}}\frac{C(x(c))}{C(x^*)}=\frac{(d+1)^{(d+1)/d}}{(d+1)^{(d+1)/d}-(d-1)}\rightarrow\infty\mbox{ as }d\rightarrow\infty.$$
\end{example}

\section{Discussion}
We consider a network pricing game in which, in the first stage, edge owners set prices so as to maximize profit, and, in the second stage, users choose paths that minimize their total costs. The problem with these games is that Nash equilibria might not exist, and if they exist, they can induce arbitrarily inefficient flows. We therefore allow for competition regulation and consider a (uniform) price cap regulation policy. Our main goal is, firstly, to find a price cap that minimizes the inefficiency of the induced flow, and, secondly, to quantify the loss in efficiency due to competition even in the presence of competition regulation. Our main results are the following. For parallel link networks with affine latency functions and a full support Wardrop equilibrium, we provide an algorithm that finds the optimal price cap in polynomial time. Due to multiplicty of Nash equilibria, the algorithm is not valid for instances without a full support Wardrop equilibrium. Then we show that the ratio between the congestion costs at an optimal price cap and the optimal congestion costs are at most 2 for duopoly instances with an uncapped Nash equilibrium. This bound lowers down to $8/7$ for affine latency functions. However, due to the nonexistence of Nash equilibria, we are able to construct a sequence of instances such that the performance of the induced flow at an optimal price cap is arbitrarily bad.

The following questions remain open. First, is there a (polynomial) algorithm that finds the optimal price cap for more general instances? In particular, for parallel link networks with affine latency functions that have no full support Wardrop equilibrium. Secondly, we have provided an instance that shows that the flow at an optimal price cap can be $8/7$ times as costly as the optimal flow. Can the upper bound in Theorem \ref{thm:main1} be extended from two link networks to arbitrary parallel link networks with affine latencies? Thirdly, for two link networks with affine latency functions, a worst-case guarantee of $8/7$ also holds for the simple algorithm that selects the best flow from the Wardrop flow and the uncapped Nash equilibrium flow. How does this simple algorithm perform for arbitrary parallel link networks? Fourthly, how do results change if we assume a different user behavior, like elastic users (Chau and Sim~\cite{ChSi03}, Hayrapetyan et al.~\cite{HaTaWe07} or Ozdaglar~\cite{Oz08}), atomic splittable users (Haurie and Marcotte~\cite{ha85}, Orda et al.~\cite{Or93} or Cominetti et al.~\cite{Co09}), or stochastic users (Daganzo and Sheffi~\cite{DaSh77} or Guo et al.~\cite{GuYaLi10})?

\section*{References}

\bibliographystyle{plain}
\bibliography{mybibfile}

\begin{thebibliography}{10}

\bibitem{AcOz07}
D.~Acemoglu and A.~Ozdaglar.
\newblock Competition and efficiency in congested markets.
\newblock {\em Mathematics of Operations Research}, 32(1):1--31, 2007.

\bibitem{AcOz:07}
D.~Acemoglu and A.~Ozdaglar.
\newblock Competition in parallel-serial networks.
\newblock {\em IEEE Journal on Selected Areas in Communications},
  25(6):1180--1192, 2007.

\bibitem{AnSe15}
E.~Anshelevich and S.~Sekar.
\newblock Price competition in networked markets: How do monopolies impact
  social welfare?
\newblock {\em International Conference on Web and Internet Economics}, pages
  16--30, 2015.

\bibitem{BeMcWi56}
M.~Beckmann, C.~McGuire, and C.~Winsten.
\newblock {\em Studies in the economics of transportation}.
\newblock Yale University Press, 1956.

\bibitem{Be63}
C.~Berge.
\newblock {\em Topological Spaces: Including a treatment of multi-valued
  functions, vector spaces, and convexity}.
\newblock Courier Corporation, 1963.

\bibitem{BoSaSc11}
V.~Bonifaci, M.~Salek, and G.~Sch{\"a}fer.
\newblock Efficiency of restricted tolls in non-atomic network routing games.
\newblock {\em International Symposium on Algorithmic Game Theory}, pages
  302--313, 2011.

\bibitem{BrHoKr12}
M.~Briest, P.~Hoefer and P.~Krysta.
\newblock Stackelberg network pricing games.
\newblock {\em Algorithmica}, 62(3-4):733--753, 2012.

\bibitem{CaChKaKrPrVo17}
I.~Caragiannis, X.~Chatzigeorgiou, P.~Kanellopoulos, G.~Krimpas, N.~Protopapas,
  and A.~Voudouris.
\newblock Efficiency and complexity of price competition among single-product
  vendors.
\newblock {\em Artificial Intelligence}, 248:9--25, 2017.

\bibitem{ChSi03}
C.~Chau and K.~Sim.
\newblock The price of anarchy for non-atomic congestion games with symmetric
  cost maps and elastic demands.
\newblock {\em Operations Research Letters}, 31(5):327--334, 2003.

\bibitem{ChNi09}
S.~Chawla and F.~Niu.
\newblock The price of anarchy in {B}ertrand games.
\newblock {\em ACM Conference on Electronic Commerce}, pages 305--314, 2009.

\bibitem{ChRo08}
S.~Chawla and T.~Roughgarden.
\newblock Bertrand competition in networks.
\newblock {\em International Symposium on Algorithmic Game Theory}, pages
  70--82, 2008.

\bibitem{CoDoRo03}
R.~Cole, Y.~Dodis, and T.~Roughgarden.
\newblock Pricing network edges for heterogeneous selfish users.
\newblock {\em Symposium on Theory of Computing}, 3(72):444--467, 2003.

\bibitem{Co09}
R.~Cominetti, J.~Correa, and N.~Stier-Moses.
\newblock The impact of oligopolistic competition in networks.
\newblock {\em Operations Research}, 57(6):1421--1437, 2009.

\bibitem{CoGuLiNiSc18}
J.~Correa, C.~Guzm\'{a}n, T.~Lianeas, E.~Nikolova, and M.~Schr\"{o}der.
\newblock Network pricing: How to induce optimal flows under strategic link
  operators.
\newblock {\em ACM Conference on Economics and Computation}, pages 375--392,
  2018.

\bibitem{CorSS04}
J.~Correa, A.~Schulz, and N.~Stier-Moses.
\newblock Selfish routing in capacitated networks.
\newblock {\em Mathematics of Operations Research}, 29(4):961--976, 2004.

\bibitem{DaSp69}
S.~Dafermos and F.~Sparrow.
\newblock The traffic assignment problem for a general network.
\newblock {\em Journal of Research of the National Bureau of Standards},
  73(2):91--118, 1969.

\bibitem{DaSh77}
C.~Daganzo and Y.~Sheffi.
\newblock On stochastic models of traffic assignment.
\newblock {\em Transportation Science}, 11(3):253--274, 1977.

\bibitem{DialpartI}
R.~Dial.
\newblock Minimal-revenue congestion pricing part i: A fast algorithm for the
  single-origin case.
\newblock {\em Transportation Research}, 33(B):189--202, 1999.

\bibitem{DialpartII}
R.~Dial.
\newblock Minimal-revenue congestion pricing part ii: An efficient algorithm
  for the general case.
\newblock {\em Transportation Research}, 34(B):645--665, 2000.

\bibitem{flJaMa04}
L.~Fleischer, K.~Jain, and M.~Mahdian.
\newblock Tolls for heterogeneous selfish users in multicommodity networks and
  generalized congestion games.
\newblock {\em Symposium on Foundations of Computer Science}, pages 277--285,
  2004.

\bibitem{FoKaLi15}
D.~Fotakis, D.~Kalimeris, and T.~Lianeas.
\newblock Improving selfish routing for risk-averse players.
\newblock {\em International Conference on Web and Internet Economics}, pages
  328--342, 2015.

\bibitem{Santiago16}
A.~Gonzalez.
\newblock Controversias por reducciones en la demanda en las concesiones de
  carreteras en chile.
\newblock {\em Revista de Derecho Econ\'{o}mico}, 76:61--97, 2016.

\bibitem{GuYaLi10}
X.~Guo, H.~Yang, and T.~Liu.
\newblock Bounding the inefficiency of logit-based stochastic user equilibrium.
\newblock {\em European Journal of Operational Research}, 201(2):463--469,
  2010.

\bibitem{HKKM2011}
T.~Harks, I.~Kleinert, M.~Klimm, and R.~M\"{o}hring.
\newblock Computing network tolls with support constraints.
\newblock {\em Networks}, 65(3):262--285, 2015.

\bibitem{ha85}
A.~Haurie and P.~Marcotte.
\newblock On the relationship between {N}ash-{C}ournot and {W}ardrop
  equilibria.
\newblock {\em Networks}, 15(3):295--308, 1985.

\bibitem{HaTaWe07}
A.~Hayrapetyan, E.~Tardos, and T.~Wexler.
\newblock A network pricing game for selfish traffic.
\newblock {\em Distributed Computing}, 19(4):255--266, 2007.

\bibitem{HoOlSk08}
M.~Hoefer, L.~Olbrich, and A.~Skopalik.
\newblock Taxing subnetworks.
\newblock {\em International Workshop on Internet and Network Economics}, pages
  286--294, 2008.

\bibitem{JeKlSc14}
T.~Jelinek, M.~Klaas, and G.~Sch\"afer.
\newblock Computing optimal tolls with arc restrictions and heterogeneous
  players.
\newblock {\em International Symposium on Theoretical Aspects of Computer
  Science}, pages 433--444, 2014.

\bibitem{Jo10}
R.~Johari, G.~Weintraub, and B.~Van~Roy.
\newblock Investment and market structure in industries with congestion.
\newblock {\em Operations Research}, 58(5):1303--1317, 2010.

\bibitem{Ka41}
S.~Kakutani.
\newblock A generalization of {B}rouwer's fixed point theorem.
\newblock {\em Duke Mathematical Journal}, 8(3):457--459, 1941.

\bibitem{KaKo04}
G.~Karakostas and S.~Kolliopoulos.
\newblock Edge pricing of multicommodity networks for heterogeneous selfish
  users.
\newblock {\em Foundations of Computer Science}, 4:268--276, 2004.

\bibitem{KlSc16}
P.~Kleer and G.~Sch{\"a}fer.
\newblock The impact of worst-case deviations in non-atomic network routing
  games.
\newblock {\em International Symposium on Algorithmic Game Theory}, pages
  129--140, 2016.

\bibitem{KoPa99}
E.~Koutsoupias and C.~Papadimitrou.
\newblock Worst-case equilibria.
\newblock {\em Annual Symposium on Theoretical Aspects of Computer Science},
  pages 404--413, 1999.

\bibitem{LaMaSa98}
P.~Labb\'{e}, M.~Marcotte and G.~Savard.
\newblock A bilevel model of taxation and its application to optimal highway
  pricing.
\newblock {\em Management Science}, 44(12):1608--1622, 1998.

\bibitem{Or93}
A.~Orda, R.~Rom, and N.~Shimkin.
\newblock Competitive routing in multiuser communication networks.
\newblock {\em IEEE/ACM Transactions on Networking}, 1(5):510--521, 1993.

\bibitem{Oz08}
A.~Ozdaglar.
\newblock Price competition with elastic traffic.
\newblock {\em Networks}, 52(3):141--155, 2008.

\bibitem{PaVa10}
C.~Papadimitrou and E.~Valiant.
\newblock A new look at selfish routing.
\newblock {\em Innovations in Computer Science}, pages 178--187, 2010.

\bibitem{Pi20}
A.~Pigou.
\newblock {\em The economics of welfare}.
\newblock Macmillan, 1920.

\bibitem{Porter93}
R.~Porter and J.~Zona.
\newblock Detection of bid rigging in procurement auctions.
\newblock {\em Journal of Political Economy}, 101(3):518--538, 1993.

\bibitem{Ro03}
T.~Roughgarden.
\newblock The price of anarchy is independent of the network topology.
\newblock {\em Journal of Computer and System Sciences}, 67(2):341--364, 2003.

\bibitem{RoTa02}
T.~Roughgarden and E.~Tardos.
\newblock How bad is selfish routing?
\newblock {\em Journal of the ACM}, 49(2):236--259, 2002.

\bibitem{Singapur10}
{Singapur Government}.
\newblock Electronic road pricing.
\newblock Technical report, Singapur Government, 2010.
\newblock \url{http://www.lta.gov.sg/motoring_matters/index_motoring_erp.htm}.

\bibitem{Stockholm07}
{Swedish Road Administration}.
\newblock Congestion tax in {S}tockholm.
\newblock Technical report, Swedish Road Administration, 2007.

\bibitem{London08}
{Transport for London}.
\newblock Central {L}ondon congestion charging.
\newblock Technical report, Transport for London, 2008.
\newblock
  \url{http://www.tfl.gov.uk/assets/downloads/sixth-annual-impacts-monitoring-report-2008-07.pdf}.

\bibitem{Bergen07}
A.~Vingan, L.~Fridstrom, and K.~Johansen.
\newblock Congestion charging in {B}ergen and {T}rondheim an alternative 20
  years ahead?
\newblock Technical report, Norwegian Transportation Institute, 2007.

\bibitem{Wan16}
C.~Wan.
\newblock Strategic decentralization in binary choice composite congestion
  games.
\newblock {\em European Journal of Operational Research}, 250(2):531--542,
  2016.

\bibitem{Wa52}
J.~Wardrop.
\newblock Some theoretical aspects of road traffic research.
\newblock {\em Proceedings of the Institute of Civil Engineers}, 1(3):325--362,
  1952.

\bibitem{XiYaHa07}
F.~Xiao, H.~Yang, and D.~Han.
\newblock Competition and efficiency of private toll roads.
\newblock {\em Transportation Research Part B: Methodological}, 41(3):292--308,
  2007.

\bibitem{YaHu04}
H.~Yang and H.-J. Huang.
\newblock The multi-class, multi-criteria traffic network equilibrium and
  systems optimum problem.
\newblock {\em Transportation Research Part B: Methodological}, 38(1):1--15,
  2004.

\end{thebibliography}

\end{document}